\documentclass[11pt,a4paper,reqno]{amsart}
\usepackage{amsmath}
\date{}
 \usepackage{amssymb,latexsym, esint}
 \usepackage{amsthm}
 \usepackage[latin1]{inputenc}
\usepackage[dvips]{graphicx}
\usepackage{color}

\newtheorem{defi}{Definition}[section]
\newtheorem{thm}{Theorem}[section]

\newtheorem{rem}{Remark}[section]

 \newtheorem{prop}{Proposition}[section]
\newtheorem{lemma}{Lemma}[section]
\newtheorem{cor}{Corollary}[section]

\def\supp{\mathop{\rm supp}}
\newcommand{\tr}{{\rm tr}}

\newcommand{\R}{\mathbb R}

\newcommand{\F}{F}

\begin{document}

\title[Agmon on Graphs]
{On Agmon metrics and exponential localization for quantum graphs}

\author{Evans M. Harrell II}
\address{School of Mathematics,
Georgia Institute of Technology,\
Atlanta GA 30332-0160, USA.} \email{ harrell@math.gatech.edu}

\author{Anna V. Maltsev}
\address{Department of Mathematics,
University of Bristol,\
Bristol BS81SD, UK.} \email{annavmaltsev@gmail.com}

\thanks{}

\begin{abstract}
We investigate the rate of
decrease at infinity of eigenfunctions of quantum graphs by using Agmon's method to
prove $L^2$ and $L^\infty$ bounds on the product of an eigenfunction with the exponential of a certain metric.
A generic result applicable to all graphs is that the exponential rate of decay
is controlled by an adaptation of the standard estimates for a line, which are of
classical Liouville-Green
(WKB) form.  Examples reveal that this estimate can be the best possible, but that a more rapid rate of decay
is typical when the graph has additional structure.  In order to understand this fact, we present two alternative estimates under more restrictive assumptions on the graph structure
that pertain to a more rapid decay.
One of these depends on
how the
eigenfunction is distributed
along a particular chosen path, while the other applies to an average of the
eigenfunction over
edges at a given distance from the root point.
\end{abstract}

\keywords{Quantum graph, Agmon metric}

\maketitle

\section{Introduction}

The goal of this article is to study the rate of decrease of eigenfunctions on infinite quantum graphs
$(\Gamma, H)$, as the distance (arc length along edges) from a designated root
point $0$
goes to infinity.
In fact our results apply more generally to any
$L^2 \cap AC^1$
solution of the
eigenvalue equation
\begin{equation}\label{EVP}
H \psi := - \psi^{\prime\prime} + V(x) \psi = E \psi
\end{equation}
on infinite subgraphs $\Gamma_0 \subset \Gamma$,
which can be disconnected from $\Gamma$
by the removal of a compact subset of the graph $S$,
and on
which
the requisite conditions
are satisfied at the vertices,
without regard to what happens on other portions of the full graph
$\Gamma$.
(For us a subgraph $\Gamma_0$ is taken to consist of a subset
$\mathcal{V}_0$ of the vertices $\mathcal{V}(\Gamma)$ along
with a subset of edges $\mathcal{E}(\Gamma)$ connected to vertices in $\mathcal{V}_0$.
A compact subset is
closed and contained within the union of a finite set of vertices and edges, with finite total edge length.)
For brevity we shall refer to
solutions of \eqref{EVP} on infinite subgraphs
as \textbf{exterior eigenfunctions}.

The philosophy introduced by Agmon \cite{Agm}
for proving exponential localization of (exterior) eigenfunctions of
equations like $(-\Delta + V({\bf x})) \psi({\bf x}) = E \psi({\bf x})$ on
$\R^d$ is to introduce an
\textbf{Agmon multiplier} $F({\bf x})$
in terms of which any
solution $\psi \in L^2$ on an exterior domain
$\mathcal{D}$ must
additionally satisfy $F \psi \in L^2(\mathcal{D})$.
We in fact go beyond this
and show that $F \psi$ is finite in Sobolev norm.
The canonical case is when $F = e^{\rho(0, {\bf x})}$,
or, typically, $F = e^{(1-\epsilon)\rho(0, {\bf x})}$,
for an Agmon metric $\rho(0, {\bf x})$
that tends to infinity
as ${\bf x} \to \infty$.
With some further effort such integrated bounds
establish exponential decrease in the pointwise sense.
A great success of the Agmon method was to extend to the
case of PDEs some estimates which, in sharper
{one-dimensional versions,}
date from Liouville and Green and are more widely
known as the WKB approximation in semiclassical analysis.
We refer to \cite{Agm,HiSi} for a full account of the Agmon method and
to \cite{Olv} for a definitive account of the
Liouville-Green approximation with rigorous error control.

Quantum graphs offer an interesting middle ground between the one-dimensional
and higher-dimensional situations, and our intent here is to explore how the Agmon
philosophy can be adapted to prove exponential
localization of eigensolutions of quantum graphs.
We are unaware of previous treatments of this question, although localization
on quantum graphs in different contexts has been considered in, e.g.,
\cite{Gei,HiPo} and references therein.

We recall that a quantum graph consists of a metric graph $\Gamma$ for which a one-dimensional
Schr\"odinger operator
$$
H = - \frac{d^2}{dx^2} + V(x)
$$
is applied to functions on the edges $e$, and vertex conditions connecting the values on adjacent
edges are imposed.
The Hamiltonian $H$ of the quantum graph is
defined as a self-adjoint operator on
$$
L^2(\Gamma) := \oplus_{e}L^2(e, dx)
$$
by extending a symmetric quadratic form initially given
on the Sobolev space $H^1(\Gamma)$,
which is by definition the subspace of
$\oplus_{e}H^1(e)$ consisting of functions
that are continuous at the vertices,
cf. \cite{BeKu}, Definition 1.3.6.  The norm on
$H^1(\Gamma)$ is defined by $\|\phi\|_{H^1}^2 =
\sum_{e}\int{(|\phi^\prime|^2 + |\phi|^2) dx}$.
Following \cite{BeKu,KoSc,Kuc}, we shall assume a global lower bound
{$\ell_{\rm min}$} to the length of the edges
of $\Gamma$ and choose the potential-energy function $V(x)$ to be continuous and to
satisfy
\begin{equation}\label{assum1}
V(x) > E > -\infty
\end{equation}
for some eigenparameter $E$, outside a finite part of the graph.  (Throughout this article $E$ will be treated as a fixed parameter.)
The vertex conditions we impose
will be of Kirchhoff type, {\it viz}.

\begin{defi}\label{Kdef}
For any connected subgraph $\Gamma_0 \subset \Gamma$,
$\mathcal{K}(\Gamma_0)$ will designate the set of functions
$f \in AC^1$
on all edges (considered as open sets) in $\Gamma_0$,
such that at each vertex $v$, $f$ is continuous, and
\begin{equation}\label{Keq}
\sum_{e {\sim} v} f_e^{\prime}(v^+) = 0.
\end{equation}
The notation in the summand indicates the limit
along the edge $e$ of $f^{\prime}(x)$as $x \to v$,
calculated in the outgoing orientation from $v$.  We shall dispense with the explicit reference to $\Gamma_0$ when the
context is clear.  We refer to the functions in $\mathcal{K}(\Gamma_0)$ as satisfying
{\rm Kirchhoff conditions}.
\end{defi}

In our analysis, the edges of $\Gamma_0$ will
sometimes be oriented with increasing distance from a
root point $0$, according to a metric $\rho$ on the graph.  This will be termed
the \textbf{distance orientation} according to the metric $\rho$.
If the graph contains cycles, the assignment of
edge orientation will break down for any edge containing a
$x_c$ at the same distance from $0$ by two or more distinct paths.  We refer to any
{$x_c$ connected to
$0$ by two or more distinct paths with the same distance $\rho(x_c,0)$ as a
{\em cut point}, and find it convenient to regard any edge containing an interior cut point as
a pair of distinct oriented}
edges joined by a degree-2 vertex located at $x_c$. In particular, in Definition \ref{d:ave} we use \textbf{Euclidean distance} by which we mean the Euclidean length of a path. Furthermore, in our integrals $dx$ e.g. in \eqref{e:e} and \eqref{e:e2} the orientation plays no role and the integrals could be taken with respect to arclength $ds$. However, $dx$ was chosen for ease of notation.

Given an edge orientation defined by a metric,
the Kirchhoff condition at a vertex $v$, if not also a cut point, reads
$$
\sum_{e {\sim} v, \rho(e) \ge v} f_e^{\prime}(v^+) = f_{{\rm incoming}\,e}^{\prime}(v^-).
$$
Thus each daughter edge $e_\ell$ spawned at $v$ carries forward
a certain \textbf{fraction of the incoming derivative}
\begin{equation}\label{e:pl}p_{v_\ell} := f_{e_\ell}^{\prime}(v^+)/f_{{\rm incoming}\,e}^{\prime}(v^-)
\end{equation} with $\sum_\ell{p_{v_\ell}} = 1$. In this paper we will sometimes work with a path $P$ from a root to $\infty$. In this case, only one edge spawned at a vertex will lie on the path $P$ and the subscript $\ell$ can be dropped from the notation for  $p_{v_\ell}$. The quantities $p_v$ will be crucial in Theorem \ref{t:pbound}.

We note that the Kirchhoff conditions are analogous to classical Neumann boundary conditions in that they are the boundary conditions associated by the Friedrichs extension with the quadratic form
\eqref{EVP}, {\em viz}.
\begin{equation}\label{qform}
\phi \to \sum_{e}\int_e{\left(|\phi^\prime|^2 + V(x)|\phi|^2 \right)dx}
\end{equation}
for $\phi \in H^1(\Gamma)$
According to \cite{Kuc}, Theorem 9,
when $V=0$ this extension is a
nonnegative self-adjoint operator, since the vertex operator included there vanishes.  By
\cite{RS2}, \S X.3, the same is true when we add a potential that is bounded from below.
All quantum graphs considered in this article are defined by such Friedrichs extensions.
We first prove several basic properties of exterior eigensolutions. We show that solutions exist for general graphs in Proposition \ref{p:existence}, and we offer a partitioned uniqueness result for graphs with additional structure in Corollary \ref{c:partuniq}. Furthermore, in Section \ref{s:ladder} we offer an example demonstrating that such solutions are not unique in general.

Then, implementing an Agmon-style argument,
we will show that $L^2$ solutions multiplied by a certain growing function $F$, will still have a finite Sobolev norm.
Our
key technical estimate is
the following:
\begin{prop}\label{t:Sobbd}
Let $\Gamma_0$ be a subgraph of $\Gamma$ and assume that $\liminf V(x, \psi) - E > 0$ on
$\Gamma_0 \setminus S_1$, where $S_1$ is compact.
Suppose that $\psi \in L^2(\Gamma_0)$
and that on each edge of $\Gamma_0 \setminus S_1$, $0 < \psi \in AC^1$ and
$\psi^{''} \ge (V(x, \psi) -E) \psi$.  Let a function $F$  defined on each edge be such that $0 < \F \in AC^1$ and for some $\delta > 0$,
\begin{equation}\label{e:constraint}
V(x, \psi) -E - \left(\frac{\F^\prime}{\F}\right)^2 \ge \delta.
\end{equation}
Then we have the following bound on the Sobolev norm of $F \psi$:
\begin{equation}\label{e:integSobbd}
\|F \psi\|_{H^1(\Gamma_0)}^2  \le C_1\left( \|\psi\|_{H^1(S_2)}^2 +
%\frac 1 \delta
\sum_{v \in \Gamma_0\setminus S_2}\sum_{e \in \Gamma_0 \setminus S_2, e \sim v}F^2 \psi_e \psi_e^\prime (v+)\right),
\end{equation}
where $S_2$ is a compact set such that $S_1 \subset S_2$ and $S_2 \backslash S_1$ contains no vertices, and $\psi_e^\prime (v+)$ indicates the outward derivative.
\end{prop}

We caution that although $\F \psi$ has a finite Sobolev norm, we have as yet said nothing about its continuity at the vertices, without which it may not belong to the space $H^1(\Gamma)$.
We shall construct three different possibilities for $F$. In all of them, we ensure that the second term on the right side of (\ref{e:integSobbd}) vanishes. Our aim is to construct an
$\F$ that grows as rapidly as possible under the constraint (\ref{e:constraint})
in order to obtain the best control on $\psi$.

In our first estimate, we show that exterior
eigenfunctions exhibit at least as rapid exponential decay
estimates as is the case for the line. Since one would expect, correctly, that
the familiar one-dimensional
Liouville-Green
expressions will play a central role in extending Agmon's method to quantum graphs,
we introduce
notation for the metric that corresponds to the Liouville-Green approximation,
extending
the definition of the standard
Agmon metric to the setting of
graphs in the following way:
\begin{defi}
Let the
\textbf{classical action metric}
from $0$ to $x$
be given by
\begin{equation}
\rho_a(y,x;E) := \min_{{\rm \, paths} \, P\, y \,{\rm to}\, x} \int_{P} (V(t) - E)_+^{1/2}dt.
\end{equation}
For simplicity we usually
set $y= 0$ for a designated root point and when clear from the context
write $\rho_a(x;E)$ or $\rho_a(x)$ in place of $\rho_a(0,x;E)$.
\end{defi}

Our first application of Proposition \ref{t:Sobbd}
shows that the classical action estimate applies universally to quantum graphs as a
bound on the Sobolev norm.  (In fact the improvement from
an $L^2$ estimate to a Sobolev estimate in the following also goes through in the
classical cases \cite{Agm,HiSi}, but this is not widely
remarked upon.)

\begin{thm} \label{t:1D}
Suppose that $\Gamma_0 \subset \Gamma$ is
a connected, infinite subgraph on which
$\liminf(V(x)- E) > 0$.
If
$\psi \in L^2(\Gamma) \cap \mathcal{K}(\Gamma_0)$ satisfies
$$
- \psi^{\prime\prime} + V(x) \psi = E \psi
$$
on the edges of $\Gamma_0$,
then for any $\delta < \liminf(V-E)$,
\begin{equation}\label{L2est}
e^{\rho_a(x;E-\delta)} \psi \in H^1(\Gamma_0){\cap L^\infty(\Gamma_0)}.
\end{equation}
\end{thm}

Generically this decay estimate cannot be improved. In Section \ref{s:ladder}, we discuss the ladder graph which has an exterior  eigenfunction that exhibits the decay in \eqref{L2est}. However, other examples
indicate that more rapid decay is in some circumstances
typical. Our motivating example is the tree or, more generally, a graph which is composed of a union of trees outside of a set of compact support. The illustrative case of a regular tree with equal lengths and equal branching numbers is treated in Section \ref{s:regtree}. In particular, we see in Section \ref{s:regtree} that for a regular tree the exterior eigenfunction is in $L^1$ and not just $L^2$. We generalize this phenomenon in the next theorem, in which we make the crucial assumption that the exterior eigenfunction has a negative derivative. We show in Corollary \ref{c:decrease} that a graph which is a union of trees outside of a set of compact support satisfies this assumption.
We introduce an Agmon metric adapted to
a given path and a given eigenfunction, and this new metric exceeds $\rho_a$ by an additional contribution from the vertices.

\begin{defi}Let $P$ be a path
from the root $0$ to $x$ and let $\psi$ be an exterior eigensolution such that $\psi' < 0$. Suppose that at each vertex
$v\in P$, $p_{v}$ as in \eqref{e:pl} is the fraction of the derivative continuing down edge $e\in P$ that is adjacent to the vertex $v$.
 Let the \textbf{Agmon path metric $\rho_P$} be given by
\begin{equation}\label{e:rhoP}
\rho_P(x, {E}) = \int_P \left((V(t)-E)_+^{1/2} + \frac{1}{2}\sum_{\{v \in P: \, p_v > 0\}}\delta_{v}(t)\log (1/p_v)\right) dt,
\end{equation}
which yields an Agmon multiplier of

$$
F_P(x, {E}) = \left(\prod_{\{v \in P: \, p_v > 0\}} \sqrt{\frac{1}{p_{v}}}\right) e^{\int_P (V(t)-E)_+^{1/2} dt}.
$$

\end{defi}

Using this new version of the Agmon metric we formulate the next theorem, which captures the additional decay particular to a given path.

\begin{thm}\label{t:pbound}
Suppose that $\Gamma_0 \subset \Gamma$ is
a connected, infinite subgraph on which
$\liminf(V(x)- E) > 0$
and that
$\psi \in L^2(\Gamma)\cap \mathcal{K}(\Gamma_0)$ satisfies
$$
- \psi^{\prime\prime} + V(x) \psi = E \psi
$$
on the edges of $\Gamma_0$ and $\psi' < 0$ outside of a set of compact support.  Consider any infinite path $P \subset \Gamma_0$,
on which the fraction of the derivative exiting from a vertex $v$ is designated
$p_{v}$.
Then for any $\delta < \liminf(V-E)$,
$e^{\rho_P(x,E-\delta)}\psi\in L^2(P).$
That is,
\[\sqrt{\prod_{v \in P}\frac{1}{p_{v}}} e^{\rho_a(x, E-\delta)} \psi \in L^2(P)\cap L^\infty(P).\]
\end{thm}

At first sight the utility of this theorem could be questioned because the information it provides
about $\psi$ seems to depend on knowing $\psi$.  However, we shall show in Section \ref{s:case} that it is
sometimes possible to determine the fractions $p_{v}$ from the structure of the graph.  Moreover,
our final theorem, which is a consequence of Theorem \ref{t:pbound}, will eliminate the use of the
$p_{v_j}$ by averaging.
Specifically, for a category of regular graphs,
we shall show that
an average of an exterior eigenfunction over the edges of the same generation always
decreases more rapidly than the one-dimensional estimate of Theorem \ref{t:1D}.

Recall that a rooted tree is regular
in the sense of Naimark and Solomyak \cite{NaSo,Sol} if the vertices and edges occur in
generations at equal distances from the root $0$,
and for each $j = 0, 1, \dots$,
\begin{itemize}
\item  All vertices of the $j$-th generation have the same branching number $b_j$.
\item  The edges emanating onward from a vertex of the $j$-th generation have identical lengths.
\end{itemize}
By convention the root $0$ corresponds to $j=0$.
\medskip

We extend this definition to include certain graphs that may contain cycles, and we allow non-uniformity in the
potential energy in the following definition.
%Indeed, our results apply with some additional complications
%to graphs with a lesser degree of topological regularity, but we reserve precise details for future work.

\begin{defi}\label{d:regular}\label{d:ave}
Consider a rooted graph $\Gamma$
with the distance orientation with respect to Euclidean distance $y$.
Vertices are assumed to occur on all paths from
$0$ whenever $y = v_j$, where $v_j$ is an increasing sequence with $v_{j+1} - v_j \ge \epsilon$ for some $\epsilon > 0$.  This implies that
the edges of each generation have identical lengths, as in the case of regular trees.  In addition,
at each generation $j$,
\begin{itemize}
\item  Every vertex {at $v_j$ has the same ongoing branching number $b_j \ge 2$}.
\item  Every vertex {at $v_j$ has the same arriving branching number $a_j \ge 1$}.
\end{itemize}
Such a metric graph will be termed a \textbf{regular braided graph}; the case of a regular tree corresponds to $a_j = 1$ for all $j$.
We define a quantum graph on a regular $\Gamma$
{with the usual Kirchhoff conditions at the vertices and refer to it}
as having \textbf{regular topology}.
On a graph of regular topology, given an $L^2$ eigensolution $\psi$, we define an \textbf{averaged wave function},
depending on $y \in \R^+$,
by
\begin{equation}\label{e:Psi}
\Psi(y) := \sum_{e: \exists x \in e: {\rm dist}(0,x) = y}{\left(\prod_{j: v_j < x}\left(\frac{a_j}{b_{j}}\right) \right)\psi_e(x),}
\end{equation}
\end{defi}

\noindent
We observe that
$\sum_{e: \exists x \in e: {\rm dist}(0,x) = y}{\left(\prod_{j: v_j < y}\left(\frac{a_j}{b_{j}}\right)\right)} = 1$, thus making our use of the term ``averaged" justified. Indeed $\Psi(y)$ is the average of $\psi(x)$ over all points $x$ at a given Euclidean distance $y$ from the root.

\begin{thm}\label{t:ave}
Suppose that $\Psi$ is the averaged eigenfunction on a quantum graph with regular topology corresponding to a
solution $\psi$ of \eqref{EVP},
for which $\psi \in L^2(\Gamma)\cap\mathcal{K}$,
and that for all $x$ such that ${\rm dist}(0,x) = y$, $V(x) \ge V_m(y)$, where
$\liminf(V_m(y) - E) > 0$.
Define
\begin{equation}\label{aveh}
\F_{\rm ave}(y, E) := \left(\prod_{j: v_j < y} \sqrt{\frac{b_j}{a_j}}\right) e^{\int_0^y \sqrt{V_m(t) - E} \,dt}.
\end{equation}
Then for each
$0 < \delta < \liminf(V_m-E)$,
\[\F_{\rm ave}(y, E-\delta) \Psi \in H^1(\R^+)\cap L^\infty(\R^+).\]
\end{thm}

The rest of the paper is organized as follows. In Section \ref{s:basic}, we prove several basic facts. We first show existence of eigensolutions in general. In the case where the graph is a finite union of trees outside of a set of compact support, we also prove a limited uniqueness and the crucial fact that $\psi' < 0$. In Section \ref{s:Agmon} we offer proofs of our theorems. The final section of this article (Section \ref{s:case}) contains case studies of illustrative examples, including regular trees and trees with a lesser degree of regularity involving two lengths (see \ref{s:2lengths}). We also include a study of the ladder graph, which is not a union of trees outside of a set of compact support, and thus provides a good testing ground for the validity our theorems in general settings.

\section{Existence, uniqueness, and basic properties of eigensolutions}\label{s:basic}

In this section we collect some useful facts about eigensolutions of quantum graphs.
Since we allow $E$ to be an arbitrary real parameter, it might be asked whether $L^2$ solutions exist and, if so, whether they can be characterized with a degree of uniqueness.  We first tackle existence.

\begin{prop}  \label{p:existence}Consider a Hamiltonian on an infinite graph $\Gamma$ satisfying the assumptions in the introduction, and fix a
connected subset  of positive measure $i_0 \subset \Gamma$.
Then for any $E < \inf_{\Gamma \setminus i_0}(V)$, there exists a function $\psi(x) \in L^2(G)$ that satisfies
\begin{equation}\label{EVEQ}
- \psi^{\prime\prime} + V(x) \psi = E \psi
\end{equation}
and the Kirchhoff conditions on $\Gamma \setminus i_0$.
\end{prop}

\begin{proof}

As shown in \cite{Kuc}, the assumptions guarantee
that the Friedrichs extension of \eqref{qform} defines a
nonnegative self-adjoint operator.  We now perturb this operator by adding a potential of the form
$$
\alpha w(x),
$$
where $w \ge 0$ is a $C^\infty$ function supported in
a finite subinterval of $i_0 \cap e$ for some edge $e$ (and not identically $0$).
We first note that that $\alpha w(x)$ is a relatively form compact perturbation of $H$ and therefore
leaves its essential spectrum unchanged,
by the following fairly standard argument based on Problem XIII.39 of \cite{RS4}:  It suffices to show that
multiplication by $w$ is a compact mapping from $H^1(\Gamma)$ to its dual space
$H^{-1}(\Gamma)$.  Now, multiplication by $w$ is a bounded map
$H^{1}(\Gamma) \to H_0^{1}(I)$ for some compact interval $I$, and the latter space
is compactly embeddable in $H^{-1}(I)$, which is in turn isomorphic
to a subspace of $H^{-1}(\Gamma)$.
It follows that
$\sigma_{\rm ess}(H + \alpha w) = \sigma_{\rm ess}(H)$.

Therefore,
if $E <0$ is in the spectrum of $\sigma_{\rm ess}(H + \alpha w)$, it is an eigenvalue of finite multiplicity
(\cite{RS4}, \S XIII.4),
implying that there exists an $L^2$ eigenfunction $\psi$ solving
$$
(H + \alpha w) \psi = E \psi
$$
on $\Gamma$.  In particular \eqref{EVEQ} holds outside $i_0$.

Now, since multiplication by $w$ is a bounded operator, the spectrum depends continuously on
$\alpha$ (\cite{RS4}, \S XII.2).  Thus consider
a normalized test function $\varphi$ supported in $\supp w$ and note that
$$
\left\langle\varphi, (H + \alpha w) \varphi \right\rangle = \mathcal{E}(\varphi) + \alpha \int_{i_1}w |\varphi|^2\,dx
$$
tends continuously to $-\infty$ as $\alpha \to -\infty$.  The Rayleigh-Ritz inequality states that
$$
\inf \sigma(H + \alpha w) \le \left\langle\varphi, (H + \alpha w) \varphi \right\rangle,
$$
so by continuity, for any given $E < 0$, there exists a value of $\alpha$ for which
$$
E \in \sigma(H + \alpha w),
$$
which finishes the proof {of existence}.
\end{proof}

As a variant of standard fact, to be found for example in \cite{BeKu}, we note that exterior
eigensolutions are in $H^1(\Gamma)$:
\begin{prop}
If $\psi$ is  an $L^2$ eigensolution of \eqref{EVP} on a subgraph $\Gamma_0$ where
$V-E \ge 0$,
and $\Gamma_0$ can be disconnected from
$S:= \{x: V-E < 0\}$ by the removal of a finite number of points $x_k$, $k = 1, \dots k_{\rm max}$
then  $\psi^\prime \in L^2(\Gamma_0)$ and $\sqrt{V-E} \, \psi \in L^2(\Gamma_0)$.
\end{prop}

\begin{proof}
By integrating by parts and invoking the vertex conditions,
\eqref{EVP} implies that
$$
\int_{\Gamma_0}\left(|\psi^\prime|^2 + (V - E) \psi^2\right) dx = \sum_{k=1}^{k_{\rm max}}(\pm\psi(x_k) \psi^\prime(x_k)),
$$
showing that a finite quantity is the sum of the squared norms of $\psi^\prime$ and $\sqrt{V-E} \psi$.
\end{proof}

In Sturmian theory, the characterization of the solution set is related to unique continuation.
In particular,  if $\liminf V(x) > 0$, then there can only be a finite number of nodes for any
solution of $- \psi^{\prime\prime} + V(x) \psi = E \psi$
on a finite or infinite interval.
In contrast, there are examples of quantum graphs for which the zero set of an eigenfunction
contains intervals or an infinite number of discrete nodes (cf. the ladder example in Section \ref{s:ladder}).  However,
the following proposition
allows a generalization of
the classical statement
about the finite number of nodes:

\begin{prop}\label{p:growingpaths}
Assume that a connected infinite
quantum graph has a minimal edge length
and that outside a
compact subset $S$,
$V(x) - E > 0$.  Assume that $x_0 \notin S$ is a boundary point of the zero set $Z(\psi) := \{x: \psi(x) = 0\}$ of an exterior $L^2$ eigenfunction $\psi.$
(The point $x_0$ is either an isolated node of $\psi$ or else a vertex that abuts an edge on which $\psi$ vanishes identically.)
Then there are at least two
oriented
paths beginning at $x_0$, along each of which $|\psi(x)|$ strictly increases until the path enters $S$.
\end{prop}

\begin{proof}
{As mentioned, $x_0$ could either be a node in the interior of an edge or a vertex, and if a vertex there are two possibilities, \emph{viz}.,
it may touch an edge on which the eigenfunction is identically zero, or else it
touches at least two other
edges on which the eigenfunction is not identically zero.
In any of these cases, by the existence-uniqueness theorem for ODEs
there must be at least one
edge $\mathcal{E}_0$ leaving $x_0$ on which $\psi^\prime(x_0^+) > 0$,
and, due to the Kirchhoff condition, one edge where $\psi^\prime(x_0^+) < 0$.
(In the case where $x_0$ is a node we regard it as a vertex of degree 2 and consider the edge where it is located as two distinct edges.)}
We discuss only the case $\psi^\prime(x_0^+) > 0$, as the argument for the case where $\psi^\prime(x_0^+) < 0$
is the same with a systematic sign difference.
Therefore there is an interval in $\mathcal{E}_0$ of the form $(0,\epsilon)$ in the variable which $=$ distance from $x_0$,
for which $\psi(x) > 0$ and $\psi^\prime(x) > 0$.  Since $V-E > 0$ in $S^c$,
$\psi^{\prime\prime} > 0$, which means that $\psi(x)$ and $\psi^\prime(x)$
must increase on all of $\mathcal{E}_0$.  Call the vertex at
which $\mathcal{E}_0$ terminates $v_1$.  Because of the Kirchhoff condition, there is at least one edge $\mathcal{E}_1\ne\mathcal{E}_0$ emanating from $v_1$ such that  $\psi^\prime(v^+) > 0$.  Repeating the argument, $\psi(x)$ and $\psi^\prime(x)$ increase on all of $\mathcal{E}_1$ and on a continuing chain of edges $\mathcal{E}_k, k = 1, \dots$.  Recall that the length of each edge is bounded from below.  Such a chain can therefore not
{include an infinite number of vertices}, because in that case $\psi \notin L^2(\Gamma)$.  The remaining possibility is that the chain on which $\psi(x)$ and $\psi^\prime(x)$ increase enters $S$.
\end{proof}

\begin{cor}\label{c:decrease}
If $\Gamma$ has only a finite number of cycles, and $\psi$ is an exterior eigenfunction,
then the null set $Z(\psi)$
has at most a finite number of connected components.  Furthermore, there exists a compact set
$S$ such that on $\Gamma \setminus S$,
$\psi$ is
monotonically decreasing as a function of the distance from $S$.
\end{cor}

The proof is a straightforward consequence of the observation that at most a finite number of paths can enter $S$ and thus only a finite number of paths as described in Proposition
\ref{p:growingpaths} are possible given that the number of cycles is bounded.
We shall henceforth refer to such functions $\psi$ as
\textbf{monotonic exterior eigenfunctions}.

It is not excluded that $\psi$ may vanish identically on
certain maximal connected infinite subgraphs, but there can be at most
finitely many such subgraphs $\Gamma_Z$, and
any such $\Gamma_Z$
connects to $\Gamma \setminus \Gamma_Z$ at only a finite set of
vertices. Thus, for graphs with $V-E > 0$ and
treelike structures
outside a compact set,
there is no loss of generality in assuming that the exterior eigenfunctions are positive,
decreasing, and convex.

A further corollary of Proposition \ref{p:growingpaths} is a partitioned uniqueness theorem for
$L^2$ solutions of $H \psi = E \psi$ on trees.

\begin{cor}\label{c:partuniq}
{
Assume that a graph $\Gamma$ contains only a finite number of cycles and that
$V-E \ge 0$ outside a compact set.  Then there is a (possibly different) compact set
$S$ such that $\Gamma \setminus S$ can be partitioned into
a finite number of maximal connected
subgraphs $\{\mathcal{T}_k\}$, intersecting only at vertices, such that any
exterior eigenfunction $\psi$
is a linear combination of functions supported on exactly one of the
$\{\mathcal{T}_k\}$.  Moreover,
the solution set supported in each
$\{\mathcal{T}_k\}$
is one-dimensional.
}
\end{cor}

\begin{proof}

Suppose that there were two linearly independent $L^2$ solutions of the eigenvalue equation,
$\psi_{1,2}$, and that their supports contain an interval in common.  Since they are linearly independent, some linear combination
$$
\psi_3 = a \psi_1 - b \psi_2
$$
must change sign on $\mathcal{T}$.  However, this contradicts
Proposition \ref{p:growingpaths}, by which no solution that changes signs on $\mathcal{T}$
can belong to $L^2$.

\end{proof}

\section{Agmon estimates for quantum graphs}\label{s:Agmon}

In this section we prove decay estimates for {exterior eigenfunctions
on quantum graphs}.
{In some regards we follow the line of reasoning
laid out in the book by Hislop and Sigal \cite{HiSi},
which contains a treatment of the Agmon method in the standard case.
However, we not only adapt their argument to graphs, but
generalize it in some ways, in particular by providing
Sobolev estimates in addition to $L^2$ estimates.
The adaptation of Agmon's method to graphs
begins with analogues of two simple
integration-by-parts lemmas}
from \cite{HiSi}.
As a matter of convenience, we state
our results in the case of real functions defined on
a quantum graph.  The extension to complex solutions is
immediate, since the real and imaginary parts of a complex eigensolution are real eigensolutions.

The first lemma, replacing \cite{HiSi}, Lemma 3.6,
is an elementary identity:

\noindent
\begin{lemma}\label{HS3.6altOLD}
Suppose that $\phi$ and $\F > 0$ are real-valued functions on the metric graph $\Gamma$
such that
$\phi \ \in AC^1$ and
$\F \in AC$.
Then for any $x$ in an edge of $\Gamma$,
\begin{equation}\label{1/FLemma}
(\F \phi)^\prime \left(\frac{\phi}{\F}\right)^\prime = (\phi^\prime)^2 - \left(\frac{F^\prime}{F}\right)^2 \phi^2.
\end{equation}
Moreover, on any subgraph $\Gamma_0 \subset \Gamma$
\begin{align}\label{1/FLemma2}
\sum_{e \in \Gamma_0}\int_{e}{\F \phi \left(-\frac{d^2}{dx^2}+ V(x)-E\right)\frac{1}{\F}\phi\,dx}
= -&\sum_{v \in \Gamma_0}\sum_{e \in \Gamma_0, e \sim v} \F\phi \frac{d}{dx}\left[\frac{1}{\F}\phi \right](v^+)\nonumber\\
&+ {\sum_{e \in \Gamma_0}}
\int_{e}{|\phi^\prime|^2 + \left(V-E - \left|\frac{\F^\prime}{\F} \right|^2\right)|\phi|^2dx}.
\end{align}
The notation in \eqref{1/FLemma2} is meant to convey that the derivatives are taken in the outward sense at the vertex.
\end{lemma}

\begin{rem}
In \cite{HiSi} Lemma 3.6, $F$ is assumed to be bounded. By keeping track of boundary terms in this lemma we are able to eliminate the need for this assumption.
\end{rem}

\begin{proof}
The first identity \eqref{1/FLemma} is an easy calculation.
The other form follows by integration by parts, the result of which
is that
\begin{small}
\begin{align}\label{e:e}
\sum_{e \in \Gamma_0}\int_{e}\F \phi \left(-\frac{d^2}{dx^2}+ V(x)-E\right)\frac{1}{\F}\phi\, dx
&= -\sum_{e\in\Gamma_0} \F\phi \frac{d}{dx}\left(\frac{1}{\F}\phi \right)\bigg|_{ei}^{ef}
\nonumber\\
&+ \sum_{e \in \Gamma_0}\int_{e}{\left(|\phi^\prime|^2
+ \left(V-E - \left|\frac{\F^\prime}{\F} \right|^2\right)|\phi|^2\right)dx},
\end{align}
\end{small}

\noindent
where $ei$ is the initial vertex of the edge $e$ and $ef$ the final vertex.  The dependence on the
edge orientation in this expression is only apparent, however:  At each vertex, all derivatives
in the integrated term are summed with an inward orientation.  Thus when the integrated terms
are collected at each vertex, the result is the expression
\eqref{1/FLemma2} which does not depend on how the edges are oriented.
\end{proof}

Our second lemma replaces
\cite{HiSi} Lemma 3.7.
\begin{lemma}\label{HS3.7alt}
Let $S_1 \subset S_2$ be two compact subsets of a subgraph $\Gamma_0$, such that
$S_2 \setminus S_1$ contains no vertices.
Let $\eta \ge 0$ be a smooth function supported in
$\Gamma_0 \setminus S_1$ such that $\eta(x) = 1$ on
$\Gamma_0 \setminus S_2$. Furthermore, let
{$\psi > 0$ satisfy $\psi'' {\geq} (V(x, \psi) - E)\psi$ on each edge.}
Then for each $x \in \Gamma_0 \setminus S_1$,
\begin{equation}\label{ptwiseSobbd}
\F^2 \eta \psi \left(-\frac{d^2}{dx^2}+ V(x)-E\right)\eta \psi
 \le C_0 \chi_{\supp \eta}(x)  (\psi)^2(x)+ G^\prime(x),
\end{equation}
where $C_0$ is a finite constant and $G := {-} \frac{1}{2} ((F\psi)^2(\eta^2)')$.
\end{lemma}
\begin{rem}
The left side of \eqref{ptwiseSobbd} is the integrand in \eqref{1/FLemma} after setting
$\phi = F \eta \psi$, as we shall do in the proof of Proposition \ref{t:Sobbd}.  Furthermore, we note that $G$ as above is 0 on all vertices since $\eta$ is chosen in such a way that $\supp \eta'$ contains no vertices.
\end{rem}

\begin{proof}
Expanding the derivatives and using that $\psi$ is a subsolution, we get that
\begin{align}
\F^2 \eta \psi \left(-\frac{d^2}{dx^2}+ V(x)-E\right)\eta \psi&= -\F^2 \eta \psi (\eta'' \psi + 2 \eta' \psi' + \eta \psi'')
 + (V - E) (\F \eta \psi)^2(x)\nonumber\\
&
\le
 -\F^2 \eta \psi (\eta'' \psi + 2 \eta' \psi' ) \nonumber\\
&
\quad\quad =
-(F\psi)^2\eta \eta'' - \frac{1}{2} F^2(\eta^2)' (\psi^2)'\nonumber\\
&
\quad\quad =
-(F\psi)^2\eta \eta'' - \frac{1}{2} ((F\psi)^2(\eta^2)')' + \frac{1}{2}(F^2(\eta^2)')'\psi^2 \nonumber\\
&
\quad\quad =
\psi^2\left(\frac{1}{2}(F^2(\eta^2)')' - F^2\eta \eta''\right) - \frac{1}{2} ((F\psi)^2(\eta^2)')' .\nonumber
\end{align}
Since the first term is supported within $\supp \eta^\prime$ and $\supp \eta'$ contains no vertices, it is dominated by
$C_0 \chi_{\supp \eta}(x) \psi^2(x)$ as claimed, establishing
\eqref{ptwiseSobbd}.

\end{proof}

With these two lemmas in hand we are ready to prove our theorems, following the
philosophy of Agmon.

\begin{proof}[Proof of Proposition \ref{t:Sobbd}]
We let $\eta$ be a smoothed characteristic function such that $\eta = 0$ on $S_1$ and 1 outside
$S_2$, and set $\phi = F \eta \psi$.
{Using \eqref{1/FLemma2} and \eqref{e:constraint} we get}
\begin{small}
\begin{equation}\label{e:e2}\begin{split}
\sum_{e \in \Gamma_0}\int_{e}\F^2 \eta \psi \left(-\frac{d^2}{dx^2}+ V(x)-E\right)&\eta \psi \, dx
\geq \sum_{e \in \Gamma_0}\int_{e}\left((F \eta \psi)'\right)^2 + \delta (F \eta \psi)^2 dx
\\
&
= \sum_{e \in \Gamma_0}\int_{e}\eta^2 \left[\left((F \psi )'\right)^2 + \delta (F  \psi)^2\right] dx + \sum_{e \in \Gamma_0}\int_{e} (\eta')^2 (F\psi)^2 + 2 \eta \eta'((F\psi)^2)'
\\
&
= \sum_{e \in \Gamma_0}\int_{e}\eta^2 \left[\left((F \psi )'\right)^2 + \delta (F  \psi)^2\right] dx + \sum_{e \in \Gamma_0}\int_{e} (\eta')^2 (F\psi)^2 - (\eta \eta')'(F\psi)^2
\\
&
=
\sum_{e \in \Gamma_0}\int_{e}\eta^2 \left[\left((F \psi )'\right)^2 + \delta (F  \psi)^2\right] dx - \sum_{e \in \Gamma_0}\int_{e} \eta \eta'' (\F\psi)^2.
\end{split}\end{equation}
\end{small}
To establish an upper bound we integrate \eqref{ptwiseSobbd} over $\Gamma_0$ to get
\begin{small}
\begin{equation}\label{e:was18}
\sum_{e \in \Gamma_0}\int_{e}\eta^2 \left[\left((F \psi )'\right)^2 + \delta (F  \psi)^2\right] \,dx \leq  \sum_{e \in \Gamma_0}\int_{e} \eta \eta''(\F\psi)^2 \,dx + C_0 \chi_{\supp \eta}(x)  (\psi)^2(x)+ G^\prime(x).
\end{equation}
\end{small}
Thus,
\begin{equation}\nonumber
\sum_{e \in \Gamma_0}\int_{e}\eta^2 \left[\left((\F \psi )'\right)^2 + (\F  \psi)^2\right] \,dx \leq C_2 \|\psi\|^2_{L^2(\supp(\eta'))},
\end{equation}
where in the last line we used that $G$ is 0 on all vertices to handle the
{last term
in \eqref{e:was18}, and $C_2$ is a constant large enough to incorporate $C_0$, the
finite maximum
value of $\eta\eta''\F^2$ on the compact set $\supp(\eta^\prime)$, and
the effect of making the coefficients $1$ and $\delta$ uniform on the left side.}

To complete the proof of Proposition \ref{t:Sobbd}
on both sides we add
$$\sum_{e \in \Gamma_0}\int_{e}(1-\eta^2) \left[\left((F \psi )'\right)^2 +  (F  \psi)^2\right] \,dx,$$
which is dominated
by a constant times $\|\psi\|_{L^2(S_2)}^2 + \|\psi'\|_{L^2(S_2)}^2$, because the support of $1-\eta^2$
is contained in the compact set $S_2$.
\end{proof}

\begin{proof}[Proof of Theorem \ref{t:1D}, Step 1]
In Step 1 we establish the finite Sobolev norm of $F(x, E-\delta)\psi(x)$
with $\F(x,E) := e^{\rho_a(x,E)}$.
Since by assumption $\psi$ is an exterior eigenfunction
and $\F(x,E- \delta)$ satisfies
assumption \eqref{e:constraint} of Proposition \ref{t:Sobbd},
\eqref{e:integSobbd} follows.
Since $\F$ is continuous and $\psi$ satisfies Kirchhoff conditions
at the vertices,
the vertex contributions to \eqref{e:integSobbd} vanish, establishing that
$e^{\rho_a(x;E-\delta)} \psi \in H^1(\Gamma_0)$.

\end{proof}

Step 2 is to show that the $H^1$ bound that has been established above
implies a pointwise bound on $F \psi$.  This is immediate from the following lemma,
choosing $\phi = \F \psi$.

\begin{lemma}\label{l:ptw}
{Suppose that $\phi \in AC^1$ on the edges of $\Gamma_0$,
and that
$$
\|\phi\|_{H^1}^2 \le \infty.
$$
Then $\phi \in L^\infty(\Gamma_0)$.}
\end{lemma}

\begin{proof}
Although $\phi$ is not assumed continuous at the vertices,
being in $H^1$ on the edges
implies that $\phi$ has well-defined finite limits
as $x$ tends to a vertex along any given edge.

We now fix $x_0 \in \Gamma_0$, and choose a function $\chi$
supported in $\{x: \rm{dist}(x, x_0) \le \frac{\ell_{\rm min}}{2}\}$, $\chi \in C^1$ on all edges intersecting this set,
and continuous and equal to $1$ at $x_0$.  (The circumlocution is only necessary in case $x_0$ is a vertex.  Here if
$x_0$ is a vertex, we interpret
$\phi(x_0)$ as the limiting value along any given edge.)  The function $\chi$ is to be
chosen so that its $C^1$ norm does not depend on $x_0$.

Because of the assumption that there is a minimum
edge length, we can write $\chi \phi(x)$ as the integral of its derivative over an interval,
which we may assume without loss of generality, by choosing an orientation for $x$,
to be of the form $I = \left(x_0 - \frac{\ell_{\rm min}}{2}, x_0\right)$, finding
\begin{align}
|\phi(x_0)| &=  |\chi(x_0) \phi(x_0)| = \left|\int_{x_0 - \frac{\ell_{\rm min}}{2}}^{x_0}{(\chi(y) \phi(y))^\prime dy}\right|\nonumber\\
&= \left|\int_{x_0 - \frac{\ell_{\rm min}}{2}}^{x_0}{\left(\chi^\prime (y) \phi(y) + \chi(y)\phi^\prime(y)\right)  dy}\right|\nonumber\\
&\le \frac 1 2 \int_{x_0 - \frac{\ell_{\rm min}}{2}}^{x_0}{\left((\chi^\prime)^2(y)+ (\phi(y))^2 + (\chi(y))^2 + (\phi^\prime(y))^2\right)  dy},\nonumber\\
\end{align}
which is bounded independently of $x_0$ by the assumptions of the lemma.
\end{proof}

\begin{proof}[Proof of Theorem \ref{t:pbound}]
We again apply Proposition \ref{t:Sobbd}.  The Agmon multiplier
$e^{\rho_P(x,E-\delta)}$ has been chosen so that the boundary terms in
\eqref{e:integSobbd} vanish when $\Gamma_0$ is identified with $P$.  This establishes that
$e^{\rho_P(x,E-\delta)}\psi$ has finite $H^1$ norm and therefore finite $L^2$ norm.  (It fails,
however, to have the continuity necessary to belong to the space $H^1(P)$.

The $L^\infty$ bound follows as before by an application of Lemma \ref{l:ptw}.

\end{proof}

We now turn to the proof of Theorem \ref{t:ave}, showing that
in the case of regular braided graphs, when the number of vertices in
generation $j$ increases without bound, the
exterior eigenfunctions as defined in Definition \ref{d:ave}
decrease on average more rapidly than
the one-dimensional upper bound of Theorem \ref{t:1D}.
As a consequence of Proposition \ref{p:growingpaths} we may assume that each $\psi_e(x) > 0$, and consequently that the averaged wave function $\Psi(y) > 0$.

\begin{prop}\label{aveprops}
Let $\psi$ be an exterior eigenfunction on a regular braided graph and suppose $V(x)\geq V_m(y)$ where $\liminf(V_m(y)-E)> 0$. Then the averaged wave function $\Psi(y)$, as defined in \eqref{e:Psi}, enjoys the following properties:
\begin{enumerate}
\item  $\Psi$ is continuous and decreasing in magnitude.
% outside $S$.
\item  Except at the positions of the vertices $y=v_j$, $\Psi$ satisfies

$$\Psi^{\prime\prime} \ge \left(V_m(y) - E\right) \Psi \ge 0.$$

\item  The derivative $\Psi^\prime$ is discontinuous at $y=v_j$, decreasing in magnitude by a
factor $p_{j} = \frac{a_j}{b_j}$.
\end{enumerate}
\end{prop}

\begin{proof}
1.   $\Psi$ must decrease in magnitude
%outside $S$
as a consequence of Proposition \ref{p:growingpaths}.  To see that
$\Psi$ is continuous,
first, for any $y\in\Gamma$,  define $W_t:=
\prod\limits_{{\stackrel{\rm\scriptstyle vertices}{\rm leading\ to\ }t}}
\frac{a_v}{b_v}$, and observe
that $\sum_{t:{\rm dist}(0,t)=y}W_t=1$. Since
$\Psi(y)=\sum_{t:{\rm dist}(0,t)=y}
W_t\psi(t)$,
each time a $y$ passes a value $v_j$, a
contribution of
$\left( \prod_{k<j} \frac{a_{k}}{b_{k}}\right) \sum_{\ell=1}^{a_{j}} \psi_{j-1,\ell}(v^-)$ to $\Psi(y)$ is replaced by
\[\left( \prod_{k<j}\frac{a_{k}}{b_{k}} \right) \cdot \frac{a_j}{b_j} \sum^{b_j}_{n=1} \psi_{j,n}(v^+_k) =
\left( \prod_{k<j} \frac{a_{k}}{b_{k}} \right) \sum_{\ell=1}^{a_{j}} \psi_{j-1,\ell}(v^-)\]
 by the continuity of $\psi$.

\noindent
2.  This is clear by linearity.

\noindent
3.  When $y$ passes a value $v_j$, $a_j$ summands of the form
$\psi_{e^\prime}^\prime(v_j^-)$ are replaced by  $\sum_{e > v_j}{\frac{a_j}{b_j} \psi_e^\prime(v_j^+)} = \sum_{e^\prime < v_j}{ \psi_e^\prime(v_j^-)}$,
according to the Kirchhoff condition.
\end{proof}

\begin{proof}[Proof of Theorem \ref{t:ave}]
We note that the averaged wave function can be considered as an exterior eigenfunction on a path as in Theorem \ref{t:pbound}, where
$p_{j} = \left(\frac{a_j}{b_j}\right)$.  Although it is possible that some of the $p_{j} > 1$, this does not affect the proof of the theorem, and convergence is not at issue, because of the following argument:

The number of vertices at the $j$-th generation is $\prod_{\ell \le j}\left(\frac{b_j}{a_j} \right)$, which shows that the factor
$\prod_{\ell \le j}\left(\frac{a_j}{b_j} \right) \le 1$ and implies that if the number of vertices at the $j$-th generation is bounded below by
a function of $j$ that tends monotonically to $+\infty$, the Agmon multiplier for $\Psi$ is exponentially smaller than $e^{-\rho_a(x)}$.

Thus Theorem \ref{t:ave} is a special case of Theorem \ref{t:pbound}.  Since $\Psi$ is continuous
at the vertices as observed in Proposition \ref{aveprops}, we know that
$F_{\rm ave}(y, E-\delta) \Psi \in H^1$ and not merely in $L^2$ with finite $H^1$ norm.
\end{proof}

\section{Case studies}\label{s:case}
In this section we develop several
illuminating
examples.
We begin by reviewing the case of the most regular tree.
\subsection{The regular tree with equal lengths L and branching numbers $b$}\label{s:regtree}
We consider a tree rooted at $v_0$ which starts with one edge and splits into $b$ edges at each vertex henceforth. We are able to construct an explicit
exterior eigenfunction on such a tree.

 We will work with transfer matrices. Suppose edge $e_j$ and $e_{j+1}$ are adjacent at a vertex $v$. If $\psi_1$, $\psi_2$ is a basis of the solution space and on an edge $e_j$ the solution is $A_j \psi_1 + B_j \psi_2$ and the solution on edge $e_{j+1}$ is given by $A_{j+1} \psi_1 + B_{j+1} \psi_2$. Then a \textbf{transfer matrix T} is a matrix such that
 \[(A_{j+1} \; \; B_{j+1})^t = T (A_{j} \; \; B_{j})^t.\] We will usually take $(\psi_1 (x), \psi_2(x)) = (\cosh(kx), \sinh(kx))$ on an edge with $k = \sqrt E$.

 We take the transfer matrix at each vertex to be
\begin{equation}
T = \begin{pmatrix}
\cosh kL & \sinh kL \\
\frac{1}{b} \sinh kL & \frac{1}{b}{\cosh kL}
\end{pmatrix}
\end{equation}
Since all the transfer matrices are equal by construction, it suffices for the purpose of characterizing the $L^2$ eigenfunction to find the eigenvalues of $T$. Since $\det T = 1/b$ and $\tr T = \left( 1 + \frac{1}{b}\right)\cosh kL > 2/\sqrt b$ both eigenvalues are real. Then solving for the eigenvector will give us the initial conditions that yield the decay corresponding to $\lambda_1^n$ where $\lambda_1$ is the smaller of the two eigenvalues and $n$ is the number of vertices away from the root. We want to compare how the decay on the tree compares to the decay on the line.

The eigenvalue $\lambda_1$ is given by
\begin{equation}\begin{split}
\lambda_1 &= \left( \frac{1}{2} + \frac{1}{2b}\right) \cosh kL - \sqrt{\left(\left( \frac{1}{2} + \frac{1}{2b}\right) \cosh kL \right)^2 - 1/b}
\\
&
=\frac{1}{\left(\frac{b}{2}+\frac{1}{2}\right)\cosh kL+ \sqrt{\left(\left(\frac{b}{2}+\frac{1}{2}\right)\cosh kL\right)^2 - b}}
\\
&
< \frac{1}{b\cosh kl}
\end{split}\end{equation}
where the last inequality follows from the fact that
\[
\left(\left(\frac{b}{2}+\frac{1}{2}\right)\cosh kL\right)^2 - b > \left(\left(\frac{b}{2}-\frac{1}{2}\right)\cosh kL\right)^2.
\]
This implies that the solution we have constructed is in $L^2$ for the tree since
\[
\int_\Gamma |\phi|^2 = C \sum_n b^n \lambda_1^{2n}.
\]
If $\lambda_1 < \alpha/\sqrt{b}$ for $\alpha <1$ then the above sum converges.

To compare this to the case of the line, we immediately see that the factor of $\frac{1}{b^n}$ makes the pointwise decay faster than the case of the line, where the decay is just $e^{-kx}$. On the other hand, if a solution is to be in $L^2$ for a tree the $1/\sqrt{b}$ factor is required for convergence. However, we have a factor of $1/b$ instead, which means that even if we consider partial integrals, the decay on the tree will be faster than on the line.

\subsection{The 2-lengths tree}\label{s:2lengths}
In this subsection we find an exterior eigenfunction for a certain tree, which is
{more sophisticated than the regular tree but still can be solved explicitly and exhibits more rapid
decrease of than the general result with the classical action.  The key to this and other examples is that if one approaches exponential decay through transfer matrices,
and parameters can be adjusted so that all the matrices in a product share a common eigenvector, then the the growth properties of the full solution built upon
that eigenvector will be determined by the product of the associated eigenvalues of the transfer matrices.}

\begin{defi}
Let the \textbf{2-lengths tree} be a rooted tree which at each vertex splits into two edges with lengths $L_1$ and $L_2$.
\end{defi}
When $L_1= L_2$, we recover the regular tree with branching number 2. Similar to the regular tree, the transfer matrix at a vertex $v$ assuming that the edge terminating at $v$ has length $L_j$ will be
\begin{equation}
T_j = \begin{pmatrix}
\cosh kL_j & \sinh kL_j \\
p_j \sinh kL_j & p_j{\cosh kL_j}
\end{pmatrix}
\end{equation}
We will seek weights for the derivative fraction $p_1$, $p_2$ at the vertex $v$, so that the eigenvector corresponding to the smaller eigenvalue of $T_1$ is the same as the eigenvector corresponding to the smaller eigenvalue of $T_2$.

We introduce the following notation. Let $c_j = \cosh kL_j$, $s_j = \sinh kL_j$, $\lambda$ be the smaller eigenvalue of $T_1$, $\mu$ be the smaller eigenvalue of $T_2$, and $(1, w)$ be the eigenvector common to $T_1$ and $T_2$ associated to $\lambda$ and $\mu$ respectively. Then from the eigenvalue equations we get that
\begin{align*}
\lambda &= \frac{1}{2}\left( c_1 + p_1 c_1- \sqrt{(c_1 + p_1 c_1)^2 - 4p_1}\right) \\
\mu &= \frac{1}{2}\left( c_2 + (1-p_1) c_2- \sqrt{(c_2 + (1-p_1) c_2)^2 - 4(1-p_1)}\right).
\end{align*}
From the two eigenvector equations we get that
\[
w = \frac{\lambda - c_1}{s_1} = \frac{\mu - c_2}{s_2}.
\]
 Putting these together we obtain the following equation for $p_1$:
 \begin{equation}
 \frac{p_1c_1 - c_1 - \sqrt{(c_1 + p_1 c_1)^2 - 4p_1}}{2s_1} = \frac{(1-p_1)c_2 - c_2 - \sqrt{(c_2 + (1-p_1) c_2)^2 - 4(1-p_1)}}{2s_2}.
 \end{equation}
 We will consider the two sides of the above equation separately and look for an intersection point in $(0, 1)$.
 We notice that when $p_1 = 0$ the left side is $-c_1/s_1 < -1$ and the right side is
 \[
 \frac{-\sqrt{c_2^2 - 1}}{s_2} = -1,
 \]
while symmetrically the right is $-1$ at $p_1 = 1$ and $-c_2/s_2 < -1$ for $p_1 = 0$ which implies
by the intermediate-value theorem that there is an intersection point, yielding a solution for some value of $p_1\in(0,1)$.

\subsection{A regular tree in the sense of Naimark and Solomyak}\label{s:eqgen}
Consider next a tree with equal lengths and branching numbers at each generation as in Definition \ref{d:regular}.
Let $b_n$ be the branching number at generation $n$. It is clear by the uniqueness of the {exterior eigenfunction} that at each vertex for $j$ between 0 and $b_n$, $p_j = 1/b_n$. Suppose that $\psi$ is an {exterior eigenfunction} and at some generation for some $j_0$, $p_{j_0} \neq 1/b_n$. Then there exists $j_1$ such that $p_{j_0} \neq p_{j_1}$. However the tree is self-similar under permutation of the branches, so a composition of $\psi$ with the isometry that maps $j_0$th branch to the $j_1$th branch will yield a second distinct eigenfunction, which would be a contradiction. From this we have a complete characterization of the eigenfunction, and thus by Theorem \ref{t:pbound} and Corollary \ref{c:partuniq} we obtain that for any directed path $P$ which includes $n$ vertices
\[
\psi(x) \leq Ce^{(1 - \delta)\rho_a(x)}\prod_{k = 1}^n \frac{1}{\sqrt{b_k}}  .
\]

\subsection{{The millipede}\label{millipede}}
Consider a graph consisting of the half axis $[0,\infty)$ with additional half axes attached at each even integer position.  On the main
half-axis, called the ``body,'' we posit $V(x) = 0$, whereas on each ``leg'' of the millipede emanating from position $x=k$
we posit a potential $V(x) = -1+\delta^2$ for $\delta > 0$.   The position $0$ on each leg corresponds to the vertex.
We set $E=-1$ so that the eigenfunction satisfies
$\psi^{\prime\prime} = \psi$ on the body between integer vertices.  On the legs it satisfies $\psi^{\prime\prime} = \delta^2 \psi$.  The $L^2$ solutions are thus proportional
to $e^{-\delta x}$ on the legs, and the solutions on the body are determined by a transfer matrix which, after an elementary calculation, has the form
$$
\left(
\begin{array}{cc}
\cosh 2 & \sinh 2\\
\sinh 2 + \delta \cosh 2 & \cosh 2 + \delta \sinh 2
\end{array}
\right).
$$
The smaller eigenvalue of the transfer matrix is
\begin{align}
\cosh 2 + \frac{\delta \sinh 2}{2} - \sqrt{\left(\cosh 2 + \frac{\delta \sinh 2}{2} \right)^2-1} &= e^{-2}\left(1 - \frac{\delta}{2}\right) + 0(\delta^2)\nonumber\\
&= e^{-2- \frac{\delta}{2} + 0(\delta^2)}.
\end{align}
This implies that the $L^2$ solution along the body is of the form $e^{\left(-1-\frac{\delta}{4} + 0(\delta^2)\right) x}$ times a periodic function.

\subsection{The ladder} \label{s:ladder} In closing we present an analysis of a ``ladder'' graph (see Figure \ref{f:ladder}), which fits within the analysis of the generic Theorem \ref{t:1D}, but not some of the other
results, being infinitely multiply connected. In particular, partitioned uniqueness (Corollary \ref{c:partuniq}) no longer applies. Here we construct two possible eigensolutions.

\begin{figure}
\includegraphics[scale = 0.5]{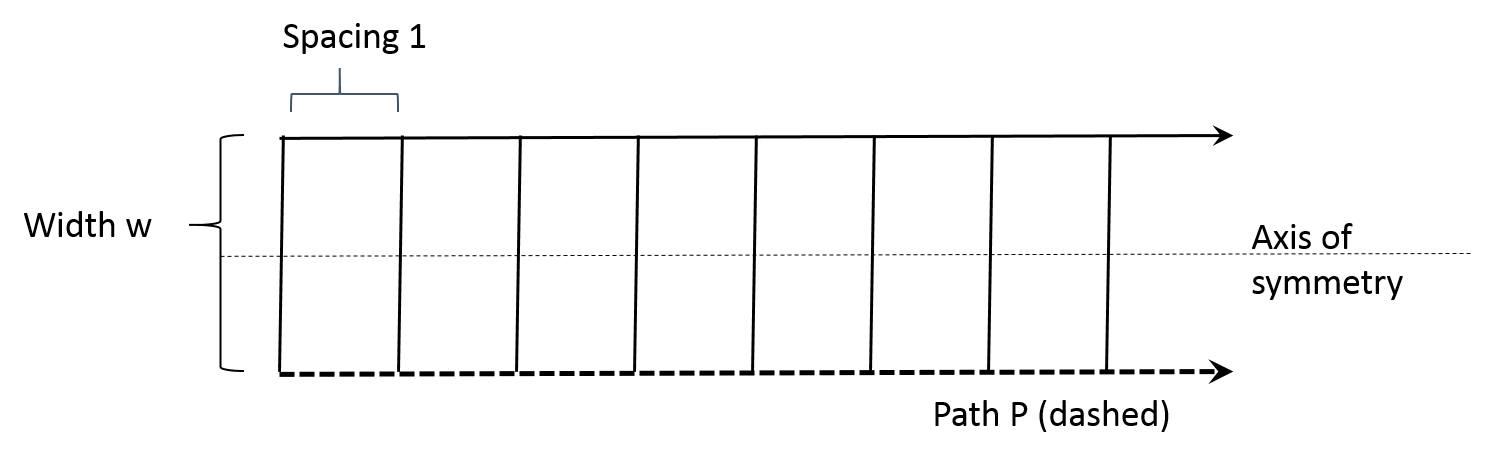}
\caption{The Ladder Graph}
\label{f:ladder}
\end{figure}

\subsubsection{Symmetric:}
 Let us take two copies of the half-line with a decaying solution. Then let us connect them at the integers with edges (rungs of the ladder) on which the solution is a constant. The result is indeed an eigensolution with the Kirchoff boundary condition at the vertices and it decays exactly at the same rate as solutions on the half-line.

\subsubsection{Antisymmetric:} One can construct a different solution on the ladder that decays faster than the above. The ladder is symmetric under reflection across a line connecting the midpoints of the rungs (the axis of symmetry is as marked in Figure \ref{f:ladder}). We assume that the rungs have length $w$ and are located at the integers. We want to construct an exterior eigenfunction which is odd under this reflection. We assume work with $V = 0$ and $E = -1$. This yields that if we parametrize each rung by $t$ with $t= 0$ in the middle, then the solution on the rung should be $\sinh t$ so that it is 0 at 0. Choosing a path along the bottom of the ladder, we compute the transfer matrix to be
$$
T = \left(
\begin{array}{cc}
\cosh 1 & \sinh 1\\
\sinh 1 + \cosh 1 \coth \frac{w}{2} & \cosh 1 + \sinh 1 \coth \frac{w}{2}
\end{array}
\right).
$$

We find that $\det T = 1$ and $\tr \, T = 2 \cosh 1 + \gamma$ where $\gamma = \coth \frac{w}{2} \sinh 1 \in (\sinh 1, \infty)$. The characteristic equation is
\[
\lambda^2 - \tr \, T \lambda + 1 = 0.
\]

If we let $\tr \,T = t$ be a parameter and differentiate the characteristic equation in $t$ we get
\[
2 \lambda \lambda ' - t \lambda' - \lambda = 0
\]
yielding that
\[
\lambda' = \frac{\lambda}{2(\lambda - \tr T/2)}
\]
which is negative for $\lambda_-$ since $\lambda_-  < \tr \, T/2$. This shows that $\lambda_-$ is monotonically decreasing in $\gamma$. One can see that $\lim_{\gamma \rightarrow \infty}\lambda_- = 0$ and $\gamma = 0$ if and only if $\lambda_- = e^{-1}$ (the 1D value). Therefore, $\lambda_- < e^{-1}$ and thus the solution satisfies a bound of the form $g e^{-|\ln \lambda_-|x}$ where $g$ is periodic and $|\ln \lambda_-|>1$.

\medskip
 \noindent
{\bf Acknowledgments}. We a grateful to the organizers of the workshop Laplacians and Heat Kernels: Theory and Applications (15w5110) at BIRS, Banff, Canada, where a lot of this work took place, and to Peter Kuchment and Gregory Berkolaiko for helpful discussions. AM acknowledges the support of the Leverhulme Trust Early Career
Fellowship (ECF 2013-613).


\begin{thebibliography}{99}

\bibitem{Agm}
S. Agmon, Lectures on exponential decay, Princeton Univ. Press, Mathematical Notes {\bf 29}.  Princeton, 1982.

\bibitem{BeKu}
G. Berkolaiko, P. Kuchment, Introduction to quantum graphs, Amer. Math. Soc. Math. Surv. Monog. {\bf 186}.  Providence, 2013.

\bibitem{BiRo}
G. Birkhoff and G.-~C Rota, Ordinary differential equations, 4th Ed., John Wiley \& Sons, New York, 1989.

%\bibitem{Boh}
%P. Bohl, \emph{\"Uber eine Differentialgleichung der St\"orungstheorie},  J. reine angew. Math. {\bf 131} (1906) 268--321.
%
%\bibitem{BoEgRu}
%\textcolor{red}{XXX NEED TO CITE THIS?}
%J. Bolte, S. Egger, and R. Rueckriemen, \emph{Heat-kernel and resolvent asymptotics for Schr\"odinger opertors on metric graphs}, arxiv:1406.1045.
%They consider only the case of a finite number of edges.

\bibitem{CoLe}
E.~A. Coddington and N. Levinson, Theory of ordinary differential equations, McGraw-Hill, New York, 1955.

%\bibitem{DaHa} E.~B. Davies and E.~M Harrell II, \emph{Conformally flat Riemannian
%metrics, Schr\"odinger operators, and semiclassical approximation}, J. Diff. Eqs.
 %\textbf{66} (1987) 165--188.

\bibitem{Gei}
L. Geisinger,
\emph{Poisson eigenvalue statistics for random Schr\"dingier operators
on regular graphs}, preprint 2014.  arXiv:1406.1608

%
% \bibitem{HaWo13}
%\textcolor{red}{XXX CITE THIS?}
%E.~M Harrell II and M.~L. Wong,
%\emph{On a transformation of Bohl and its discrete analogue}, Proc. Symp. Pure Math.  \textbf{87} (2013) 191--203.
%
%\bibitem{HaWo14}
%\textcolor{red}{XXX CITE THIS?}
%E.~M Harrell II and M.~L. Wong,
%\emph{On the behavior at infinity of solutions to difference equations in Schr\"odinger form}, Operators and Matrices  \textbf{8} (2014) 357--387.

\bibitem{HiPo}
P. Hislop and O. Post,  \emph{Anderson localization for radial tree-like random quantum graphs},
Waves in Random and Complex Media \textbf{19} (2009) 216--261.

\bibitem{HiSi}
P. Hislop and I. M. Sigal, Introduction to spectral theory, with applications to
Schr\"odinger operators, Springer Applied Mathematical Sciences 113.  New York:  Springer-Verlag, 1996.

\bibitem{KoSc}
V. Kostrykin and  R. Schrader, \emph{Kirchhoff's rule for quantum wires},
J. Phys. A: Math. Gen. \textbf{32} (1999) 595--630.

\bibitem{Kuc}
P. Kuchment, \emph{Quantum graphs: I. Some basic structure}, Waves in Random Media, \textbf{14:1} (2004) S107--S128.

\bibitem{NaSo}
K. Naimark and M. Solomyak Eigenvalue estimates for the weighted Laplacian on metric trees Proc. Lond.
Math. Soc. 80 (2000) 690--724.

\bibitem{Olv}  F.~W.~J. Olver, Asymptotics and special functions.  New York:  Academic Press, 1974.

\bibitem{RS2}
{M. Reed and B. Simon, Methods of modern mathematical physics, II. Fourier analysis, Self-adjointness, Academic Press, New York, 1975.}

\bibitem{RS4}
M. Reed and B. Simon, Methods of modern mathematical physics, IV. Analysis  of operators, Academic Press, New York, 1978.

\bibitem{Sol}
M. Solomyak On the spectrum of the Laplacian on regular metric trees Waves Random Media 14 (2004) S155Ð71.

\end{thebibliography}
\end{document}